\documentclass[a4paper, 10pt, onecolumn, oneside, article]{memoir}
\usepackage[english]{babel}
\usepackage[utf8]{inputenc}
\usepackage[T1]{fontenc}

\counterwithout{section}{chapter}
\setsecnumdepth{subsection}
\setlrmarginsandblock{2cm}{2cm}{*}
\setulmarginsandblock{3cm}{3cm}{*}
\checkandfixthelayout

\usepackage{derivative}  % defines \odif{} \odv{} \pdv{}...
\usepackage{amsmath, amssymb, amsthm, mathtools, cancel}
\usepackage{listings}  % for code listings

\usepackage{widetext}  % used for typesetting wide (single column) text across two-columns
\usepackage{multicol}  % widetext seems to not work properly with \documentclass[twocolumn, ...], so we use multicol instead

\usepackage{graphicx, float, adjustbox, xcolor}
\usepackage[font=small,labelfont=bf,textfont=it]{caption}
\usepackage{hyperref}
    \hypersetup{
        colorlinks,
        linkcolor={red!50!black},
        citecolor={blue!50!black},
        urlcolor={blue!80!black}
    }
\usepackage{cleveref}

\newtheorem{theorem}{Theorem}  % {name}{style}{number within}

\newtheorem{lemma}[theorem]{Lemma}

\theoremstyle{definition}

\newenvironment{circuit}
  {\crefalias{equation}{circuit}\begin{equation}}
  {\end{equation}\ignorespacesafterend}
\crefname{circuit}{circuit}{circuits}
\Crefname{circuit}{Circuit}{Circuits}

    \newcommand{\half}{\frac{1}{2}} %make \frac{1}{2} with \half
    \newcommand{\define}{\equiv} %definition sign with \define
    \newcommand{\tensor}{\otimes}
    
    \newcommand{\argdot}{\makebox[1em]{$\cdot$}}
    \newcommand{\idty}{\mathbb{I}}
    \newcommand{\mat}[1]{\begin{bmatrix*}[r]#1\end{bmatrix*}}
    \newcommand{\fourier}{\mathcal{F}}

    \newcommand{\qq}[1]{\quad\text{#1}\quad}
    \newcommand{\vb}[1]{\boldsymbol{\mathrm{#1}}}

    \AtBeginDocument{%

    }

    \newcommand{\qunaught}{\varnothing}

\NewDocumentCommand{\braket}{ o m m o }{%
    \ensuremath{%
        \IfValueT{#1}{{}_{#1}} % Prescript
        \langle #2 | #3 \rangle
        \IfValueT{#4}{_{#4}} % Subscript
    }%
}
\NewDocumentCommand{\bra}{ o m }{%
    \ensuremath{%
        \IfValueT{#1}{{}_{#1}} % Prescript
        \langle #2 |
    }%
}
\NewDocumentCommand{\ket}{ m o }{%
    \ensuremath{%
        | #1 \rangle
        \IfValueT{#2}{_{#2}} % Subscript
    }%
}

\usepackage{tikz}
\usetikzlibrary{quantikz2}  % library for drawing quantum circuits (uses tikz implicitly).
\newcommand{\qtikzdagger}[1]{\wire[r][1]["\dagger"{above,pos=0.1}, draw=none]{q}}

\usetikzlibrary{calc,arrows.meta}
\newlength{\bswidth}
\newlength{\bsrad}
\tikzset{BSup/.style={
    draw=none,
    minimum width=1.2cm,
    minimum height=1cm,
    path picture={
        \begin{scope}
        \draw[rounded corners=\bsrad] (path picture bounding box.center) +(-0.6, 0.5) -- +(-\bswidth, 0.5) -- +(\bswidth, -0.5) -- +(0.6, -0.5);
        \draw[rounded corners=\bsrad] (path picture bounding box.center) +(-0.6, -0.5) -- +(-\bswidth, -0.5) -- +(\bswidth, 0.5) -- +(0.6, 0.5);
        \filldraw[color=black, fill=white] (path picture bounding box.center) +(-0.4, -0.05) rectangle +(0.4, 0.05);
        \draw[thin, arrows = {-Stealth[inset=1pt, length=3pt, angle'=45, round]}] (path picture bounding box.center) +(0,-0.3) -- +(0,0.3);
        \end{scope}
    },
}}
\tikzset{BSdown/.style={
    draw=none,
    minimum width=1.2cm,
    minimum height=1.2cm,
    path picture={
        \begin{scope}
        \draw[rounded corners=\bsrad] (path picture bounding box.center) +(-0.6, 0.5) -- +(-\bswidth, 0.5) -- +(\bswidth, -0.5) -- +(0.6, -0.5);
        \draw[rounded corners=\bsrad] (path picture bounding box.center) +(-0.6, -0.5) -- +(-\bswidth, -0.5) -- +(\bswidth, 0.5) -- +(0.6, 0.5);
        \filldraw[color=black, fill=white] (path picture bounding box.center) +(-0.4, -0.05) rectangle +(0.4, 0.05);
        \draw[thin, arrows = {Stealth[inset=1pt, length=3pt, angle'=45, round]-}] (path picture bounding box.center) +(0,-0.3) -- +(0,0.3);
        \end{scope}
    },
}}
\newcommand{\bsup}[1]{\gate[2, style={BSup}]{\hspace{5mm}}}  % beamsplitters only align if [row sep={1cm, between origins}] is given to quantikz
\newcommand{\bsdown}[1]{\gate[2, style={BSdown}]{\hspace{5mm}}}  % must be provided after quantikz2

\usepackage{csquotes} % recomended by biblatex - don't know why
\usepackage[style=numeric-comp, autocite=inline, giveninits=true, sorting=none, url=false, isbn=false, maxbibnames=2]{biblatex}
    \DeclareNameAlias{labelname}{given-family}  % makes \citeauthor include given (first) name as well as family name
\title{Performance analysis of GKP error correction}

\usepackage{authblk}
\author[1,2]{Frederik K. Marqversen}
\author[2]{Janus H. Wesenberg}
\author[1,2]{Nikolaj T. Zinner}
\author[3]{Ulrik L. Andersen}

\affil[1]{Department of Physics and Astronomy, Aarhus University, DK-8000 Aarhus C, Denmark}
\affil[2]{Kvantify ApS, DK-2300 Copenhagen S, Denmark.}
\affil[3]{Center for Macroscopic Quantum States (bigQ), Department of Physics, Technical University of Denmark, Fysikvej, 2800 Kgs., Lyngby, Denmark}

\date{May 2025}

\begin{document}
\maketitle
\begin{abstract}
   Quantum error correction is essential for achieving fault-tolerant quantum computing. Gottesman-Kitaev-Preskill (GKP) codes are particularly effective at correcting continuous noise, such as Gaussian noise and loss, and can significantly reduce overhead when concatenated with qubit error-correcting codes like surface codes. GKP error correction can be implemented using either a teleportation-based method, known as Knill error correction, or a quantum non-demolition-based approach, known as Steane error correction. In this work, we conduct a comprehensive performance analysis of these established GKP error correction schemes, deriving an analytical expression for the post-correction GKP squeezing and displacement errors. Our results show that there is flexibility in choosing the entangling gate used with the teleportation-based Knill approach. Furthermore, when implemented using the recently introduced qunaught states, the Knill approach not only achieves superior GKP squeezing compared to other variants but is also the simplest to realize experimentally in the optical domain.
\end{abstract}
\vspace{1em}

\begin{multicols}{2}

\section{Introduction} \label{sec:introduction}

Efficient quantum error correction is essential for mitigating errors in quantum computing and enabling the scaling of quantum hardware to practical, fault-tolerant systems capable of executing advanced quantum algorithms, such as Shor’s algorithm~\cite{shor_algorithms_1994}. Traditional quantum error correction encodes logical qubits in a collection of physical two-level systems~\cite{gottesman_introduction_2009}. An alternative approach leverages bosonic codes, where logical qubits are encoded in multi-level bosonic systems, with prominent examples including two-component cat states and Gottesman-Kitaev-Preskill (GKP) states~\cite{gottesman_encoding_2001}. Bosonic GKP codes are particularly well-suited for correcting continuous noise, such as Gaussian noise and loss, and when concatenated with qubit error-correcting codes, like surface codes, they enable fault-tolerant quantum computing~\cite{menicucci_fault-tolerant_2014,larsen_fault-tolerant_2021,tzitrin_fault-tolerant_2021,noh_fault-tolerant_2020}. A unique advantage of GKP codes lies in their analog syndrome information~\cite{fukui_analog_2017}, which can significantly improve the decoders of concatenated qubit error-correcting codes, thereby reducing computational overhead~\cite{noh_fault-tolerant_2020,noh_low-overhead_2022,tzitrin_fault-tolerant_2021}. This feature positions GKP codes as particularly promising candidates for fault-tolerant quantum computing. Consequently, their implementation has been explored across various physical platforms, including superconducting circuit QED~\cite{sivak_real-time_2023,ni_beating_2023}, trapped ions~\cite{de_neeve_error_2022}, and preliminary attempts in a photonic system~\cite{konno_logical_2024}.

The ideal GKP logical states $\ket{0}_L$ and $\ket{1}_L$ are defined as infinite Dirac combs with a spacing of $2\sqrt{\pi}$ as illustrated on the left-hand side of \cref{fig:gkp states}. The logical $\ket{0}_L$-state has peaks at all even multiples of $\sqrt{\pi}$, while the logical $\ket{1}_L$-state has peaks at all the odd multiples. These states are unphysical since they correspond to infinite-energy states, being composed of an infinite sum of infinitely squeezed states. Realistic GKP states, however, are constrained to finite energies, where the delta functions are replaced by Gaussian peaks of finite width (see \cref{fig:gkp states}, right). Consequently, realistic GKP states are finitely squeezed and thus inherently noisy. Nevertheless, as long as the squeezing level is sufficiently high, fault-tolerant quantum computing remains achievable. For instance, concatenating GKP error correction with surface-code qubit correction achieves fault tolerance at a threshold squeezing level of approximately 10dB~\cite{noh_low-overhead_2022}. This level of squeezing is similar to the recent experimentally obtained result in the MW regime using a circuit QED system~\cite{sivak_real-time_2023}.   

%Practical implementations of GKP qubits, however, are constrained to finite-energy and finite-width approximations \cite{gottesman_encoding_2001}. These approximations, denoted $\ket{0}_\Delta$ and $\ket{1}_\Delta$, are inherently non-ideal as illustrated on the right-hand side of \cref{fig:gkp states}. 

\begin{figure*}
    \centering
    \includegraphics[width=\linewidth]{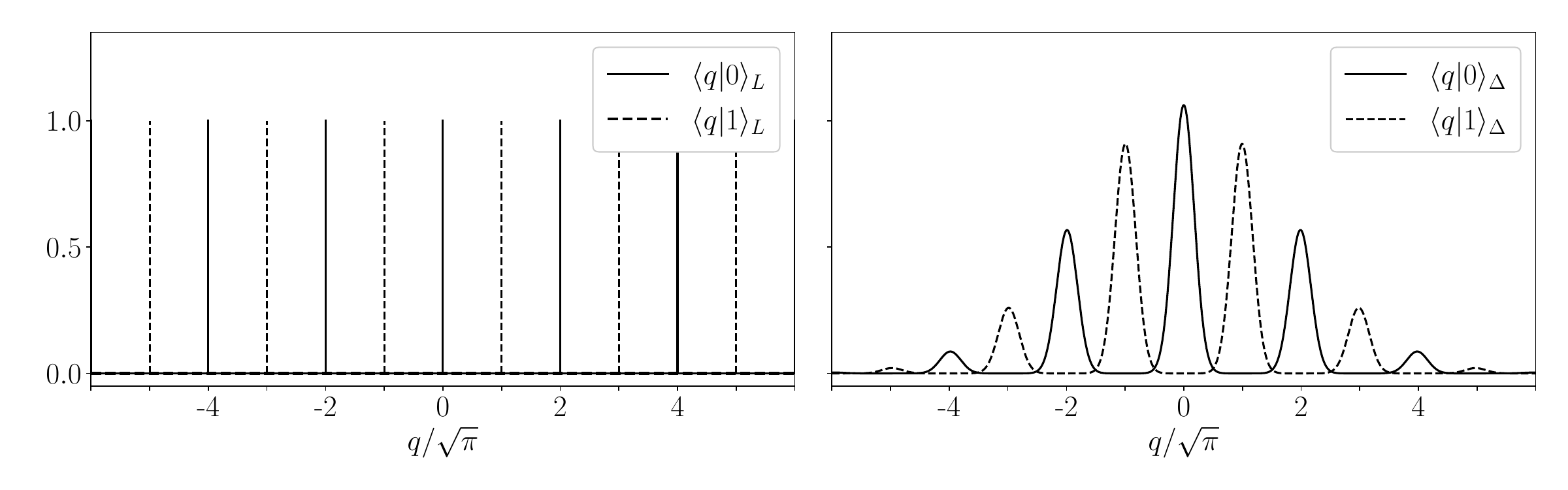}
    \caption{Wave functions in the $q$-quadrature basis $\psi(q) = \braket{q}{\psi}$. \textbf{(Left)} The Logical GKP basis states labelled by $\ket{0}_L$ and $\ket{1}_L$. \textbf{(Right)} Physical states approximating the GKP basis state labelled by $\ket{0}_\Delta$ and $\ket{1}_\Delta$ using the photon number dampening approximation $\ket{\argdot}_\Delta = e^{-\Delta \hat{n}} \ket{\argdot}_L$.}
    \label{fig:gkp states}
\end{figure*}

Currently, there are two predominant approaches to GKP error correction, both of which are measurement-based. The first approach, referred to as Knill-type error correction (see \cref{fig:Knill circuits}), is based on quantum teleportation \cite{gottesman_demonstrating_1999} and was originally developed for general stabiliser codes \cite{knill_scalable_2005} before being adapted for the GKP code~\cite{tzitrin_progress_2020, walshe_continuous-variable_2020}. The standard implementation entangles the input state with a Bell state $\ket{\Phi^+}_L = \frac{1}{\sqrt{2}} (\ket{00}_L + \ket{11}_L)$ using a controlled-X (CX) gate (see \cref{fig:Knill circuits}, left). An alternative implementation, introduced in Ref. \cite{walshe_continuous-variable_2020}, replaces the CX gate with a beam splitter, as illustrated on the right-hand side of \cref{fig:Knill circuits}. While the hardware requirements for these two implementations differ significantly -- particularly in photonic systems where inline squeezing required for the CX gate can be challenging to achieve -- the two circuits are mathematically equivalent as detailed below. Due to this equivalence, both configurations are collectively referred to as the Knill approach. 

The second main approach, known as Steane-type error correction \cite{steane_error_1996}, differs fundamentally from the Knill approach by not relying on teleportation. Instead, it utilizes two consecutive quantum non-demolition measurements involving two ancilla GKP states, along with the CX gate, its inverse, and homodyne detection. The circuit implementation, illustrated on the left-hand side of \cref{fig:Steane}, was first introduced in \cite{steane_error_1996} and later refined in \cite{gottesman_encoding_2001, tzitrin_progress_2020, baragiola_all-gaussian_2019, li_correcting_2023}. Unlike Knill-type error correction, where the input state is teleported and emerges in a different mode, Steane-type error correction keeps the input state within the same mode, which can provide practical advantages depending on the specific implementation. On a more fundamental level, the Knill approach extract the error syndromes for the conjugate quadratures -- amplitude and phase -- simultaneously as a result of the Bell measurement while in the Steane approach obtains the error syndromes sequentially by first conducting a non-demolition measurement of one quadrature followed by a non-demolition measurement of the conjugate quadrature.

In this paper, we provide a comprehensive comparison of the Knill- and Steane-type GKP error correction schemes. Specifically, we demonstrate that the two teleportation-based implementations of Knill-type error correction (illustrated in \cref{fig:Knill circuits}) are mathematically equivalent and, therefore, exhibit identical performance. However, the exact performance of Knill-type error correction is tunable, depending on the choice of the entangled resource state used for teleportation. We show that, for a specific choice of the entangled resource, the Knill- and Steane-type approaches are equivalent. In contrast, for an optimal choice of the entangled resource state -- such as qunaught states~\cite{walshe_continuous-variable_2020} -- Knill-type error correction significantly outperforms Steane-type correction, assuming the same level of squeezing across all resource states. This finding establishes the Steane-type error correction scheme as a non-optimal special case of the Knill-type approach. In our analysis we also derive explicit expressions for the squeezing and displacement errors after correction.  

The paper is organized as follows. \Cref{sec:notation} introduces the foundational notation and operations of GKP states. In \cref{sec:circuit analyses}, we prove the equivalence of the two teleportation-based Knill-type circuits and establish that the Steane-type scheme is a subset of the Knill-type approach. Moreover we present the derivation of the  projection operator for the generalized Knill-type error correction scheme, which is subsequently utilized in \cref{sec:analysis of pi} to find the state after error correction in terms of mean values and variances of the GKP Gaussian peaks. This result is finally used in \cref{sec:error analysis} to evaluate the performance of the Knill-type error correction scheme using two different types of entangled states. 

%However, alternative entangling gates can also be used. In optical systems, a beam splitter is a natural choice for this purpose. Both implementations -- the CX-based and the beam-splitter based circuits -- are shown in \cref{fig:Knill circuits}. We show that these two circuits are \emph{exactly equal}.  As a result, either implementation can be referred to as the Knill approach. The main difference between the two lies in their resource requirements: the CX-based implementation requires inline squeezing, while the beam-splitter-based implementation is often preferable in practical optical setups due to its simplicity and hard-ware efficiency.

\newcommand{\knillCX}{
    \begin{quantikz}[row sep={1cm, between origins}]
        \lstick{$\ket{\psi}$}             & \targ{} \qtikzdagger{} & \meterD{\hat{q}}   & \setwiretype{c} \rstick{$s_1$} \\
        \lstick[2]{$\ket{\Phi^+}_\Delta$} & \ctrl{-1}              & \meterD{\hat{p}}   & \setwiretype{c} \rstick{$s_2$} \\
                                          & \ghost{R}              & \rstick{$\ket{\psi'}$}
    \end{quantikz}
}
\newcommand{\knillBS}{
    \begin{quantikz}[row sep={1cm, between origins}]
        \lstick{$\ket{\psi}$}             & \bsup{}   & \meterD{\hat{q}}  & \setwiretype{c} \rstick{$s_1/\sqrt{2}$} \\
        \lstick[2]{$\ket{\Phi^+}_\Delta$} &           & \meterD{\hat{p}}  & \setwiretype{c} \rstick{$s_2/\sqrt{2}$} \\
                                          & \ghost{R} & \rstick{$\ket{\psi'}$}
    \end{quantikz}
}
\begin{figure*}
    \centering
    \begin{tikzpicture}
        \node[anchor=east] at (-1, 0) {\knillCX};
        \node[anchor=west] at (0.8, 0) {\knillBS};
        \node at (0, 0) {\huge\bfseries =};
        \node[anchor=center] at (-6, 2) {\Large (Knill)};
        \node[anchor=center] at (7, 2) {\phantom{\Large (Knill)}};
    \end{tikzpicture}
    \caption{Circuits for Knill-type error correction based on using either a controlled displacement \textbf{(left)} or a beam splitter \textbf{(right)} for entangling the input with the ancillas.  As demonstrated in \cref{sec:circuit analyses}, both circuits produce identical error-corrected states, as indicated by the equality sign.}
    \label{fig:Knill circuits}
\end{figure*}
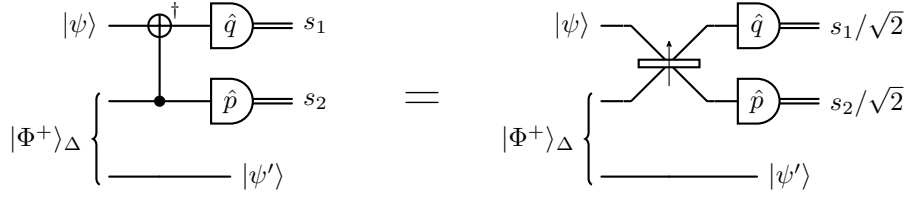

\newcommand{\steane}{
    \begin{quantikz}[row sep={1cm, between origins}]
        \lstick{$\makebox[0pt][l]{\ket{\psi}} \phantom{\ket{0}_\Delta}$}      & \ctrl{1} & \gate{F^\dagger} & \phase{}  & \gate{F}       & \gate{\hat{D}(-\vb{s})} & \rstick{$\ket{\psi'}$} \\
            \lstick{$\ket{0}_\Delta$}  & \phase{} &                        &           & \meterD{\hat{p}} & \setwiretype{c} \wire[u][1]{c} \rstick{$s_1$} \\
        \lstick{$\ket{0}_\Delta$}  & \gate{\hat{R}(\pi)} &             & \ctrl{-2} & \meterD{\hat{p}} & \setwiretype{c} \wire[u][1]{c} \rstick{$s_2$}
    \end{quantikz}
}
\newcommand{\knillassym}{
    \begin{quantikz}[row sep={1cm, between origins}]
        \setwiretype{n}           &          & \lstick{$\ket{\psi}$} & \targ{} \qtikzdagger{} \setwiretype{q}    & \meterD{\hat{q}}   & \setwiretype{c} \rstick{$s_1$} \\
        \gategroup[2,steps=3,style={dashed,rounded corners,fill=black!0, inner xsep=2pt},background,label style={label position=below,anchor=north,yshift=-0.2cm}]{$\ket{\Phi^+}_{0, \Delta}$}
        \midstick{$\ket{0}_\Delta$} & \gate{F} & \ctrl{1}              & \ctrl{-1} & \meterD{\hat{p}}   & \setwiretype{c} \rstick{$s_2$} \\
        \midstick{$\ket{0}_\Delta$} &          & \targ{}               & \ghost{R} & \rstick{$\ket{\psi'}$}
    \end{quantikz}
}
% % Alternate circuit for special Knill circuit
\newcommand{\altknillassym}{
    \begin{quantikz}[row sep={1cm, between origins}]
        \setwiretype{n}           &          &          & \lstick{$\ket{\psi}$}          & \targ{} \setwiretype{q}    & \meterD{\hat{q}=s_1}   & \setwiretype{c} \\
        \lstick{$\ket{0}_\Delta$} & \gate{F} & \ctrl{1} & \midstick[2]{$\ket{\Phi^+}_L$} & \ctrl{-1} & \meterD{\hat{p}=s_2}   & \setwiretype{c} \\
        \lstick{$\ket{0}_\Delta$} &          & \targ{}  &                                & \ghost{R} & \rstick{$\ket{\psi'}$}
    \end{quantikz}
}
% % Alternate circuit for special Knill circuit
\newcommand{\altaltknillassym}{
    \begin{quantikz}[row sep={1cm, between origins}, color=black!40]
        \setwiretype{n}           &          & \lstick{$\ket{\psi}$} & \targ{} \setwiretype{q}    & \meterD{\hat{q}=s_1}   & \setwiretype{c} \\
        \lstick[label style=black]{$\ket{0}_\Delta$} \setwiretype{n} & \gate[style={draw=black}, label style=black]{F} \wire[l][1][style={black}]{q} \wire[r][2][style={black}]{q} & \ctrl[style={black}]{1}             & \ctrl{-1} \wire[r][1][]{q} & \meterD{\hat{p}=s_2}   & \setwiretype{c} \\
        \lstick[label style=black]{$\ket{0}_\Delta$} \setwiretype{n} &          \wire[l][1][style={black}]{q} \wire[r][2][style={black}]{q} & \targ[style={black}]{}              & \ghost{R} \wire[r][1][]{q} & \rstick{$\ket{\psi'}$}
    \end{quantikz}
}
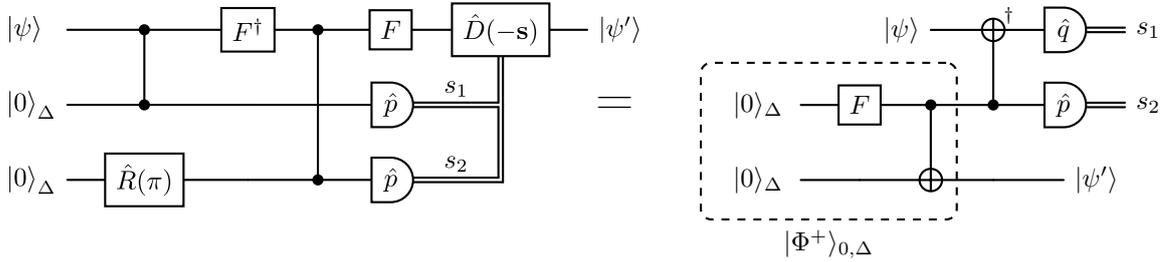
\begin{figure*}
    \centering
    \resizebox{\textwidth}{!}{
    \begin{tikzpicture}
        % \node[anchor=east] at (0.5, 0) {\steane};
        \node[anchor=east] at (0.6, 0) {\includegraphics[]{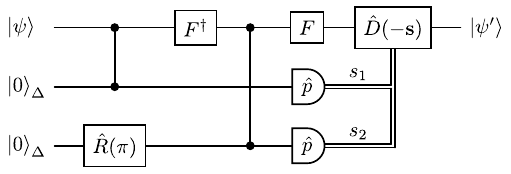}};
        \node[anchor=west] at (1, -0.375) {\knillassym};
        \node at (0, 0) {\huge\bfseries =};
        \node[anchor=center] at (-8.5, 2) {\Large (Steane)};
    \end{tikzpicture}
    }
    \caption{\textbf{(Left)}:Circuit implementation of Steane-type error correction. \textbf{(right)}: An equivalent Knill-type error correction scheme, where the Bell state is prepared using the standard qubit-based preparation circuit. As shown in \cref{sec:circuit analyses}, Steane-type error correction is a special case of Knill-type error correction.}
    \label{fig:Steane}
\end{figure*}

%Tzitrin et al. \cite{tzitrin_progress_2020} suggest that Knill-type error correction may outperform Steane-type error correction if a reliable source of high quality GKP states is available, as it reduces the number of operations directly applied to the noisy input state. However, as shown in \cref{fig:Steane}, Steane-type error correction is \emph{mathematically equivalent} to a special case of Knill-type error correction when the Bell state is prepared using the following circuit, referred to as the standard qubit preparation circuit:
% \begin{circuit}
% \begin{quantikz}[row sep={1cm, between origins}]
%     \lstick{$\ket{0}_\Delta$} & \gate{F} & \ctrl{1} & \rstick[2]{$\ket{\Phi^+}_{0,\Delta}$} \\
%     \lstick{$\ket{0}_\Delta$} &          & \targ{}  &
% \end{quantikz}
% \label{circ:standard bell}
% \end{circuit}
%This equivalence highlights that the two approaches are more similar than they might appear initially. Additionally, for performance analyses, Steane-type error correction can often be treated as a special case of the Knill-type error-correction.

\begin{figure*}
    \centering
    \includegraphics[width=0.4\linewidth]{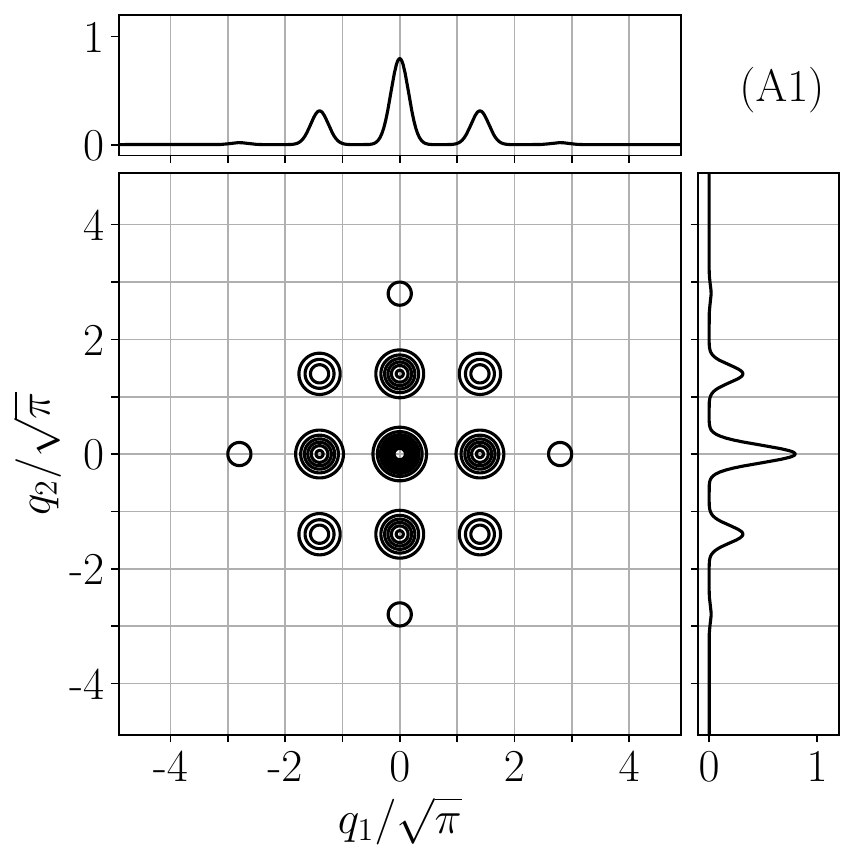}
    \includegraphics[width=0.4\linewidth]{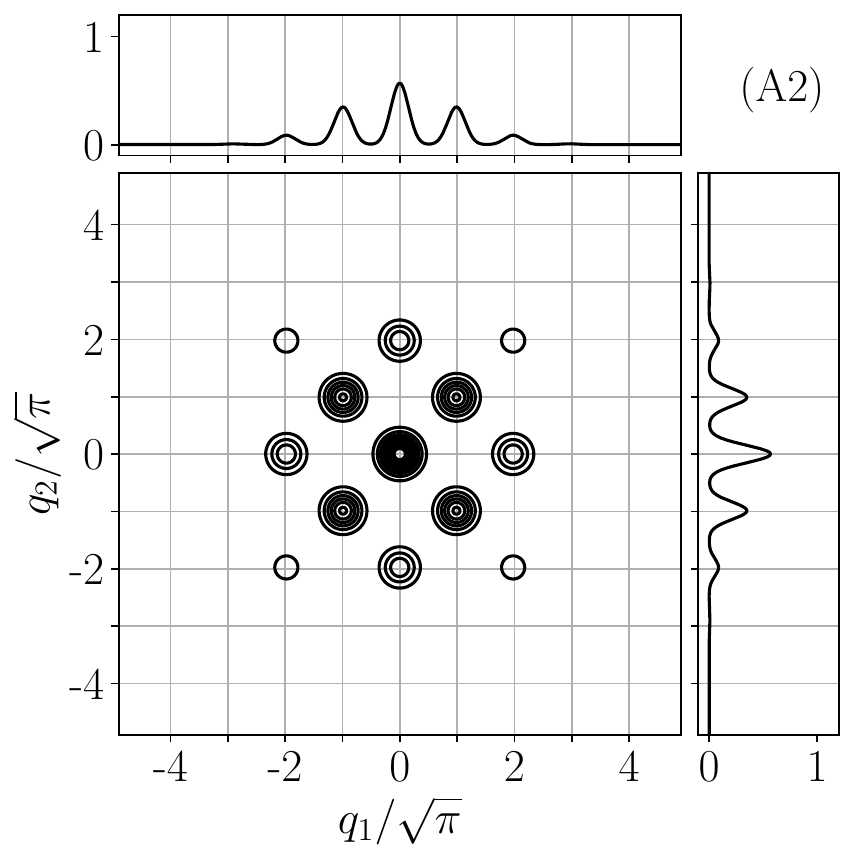}
    \includegraphics[width=0.4\linewidth]{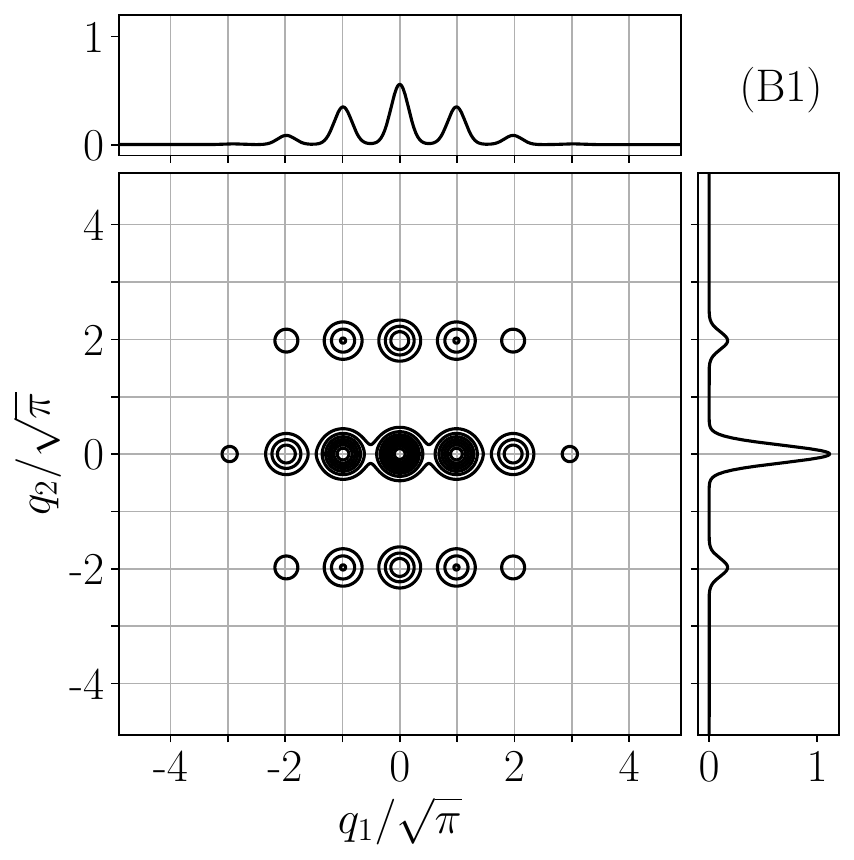}
    \includegraphics[width=0.4\linewidth]{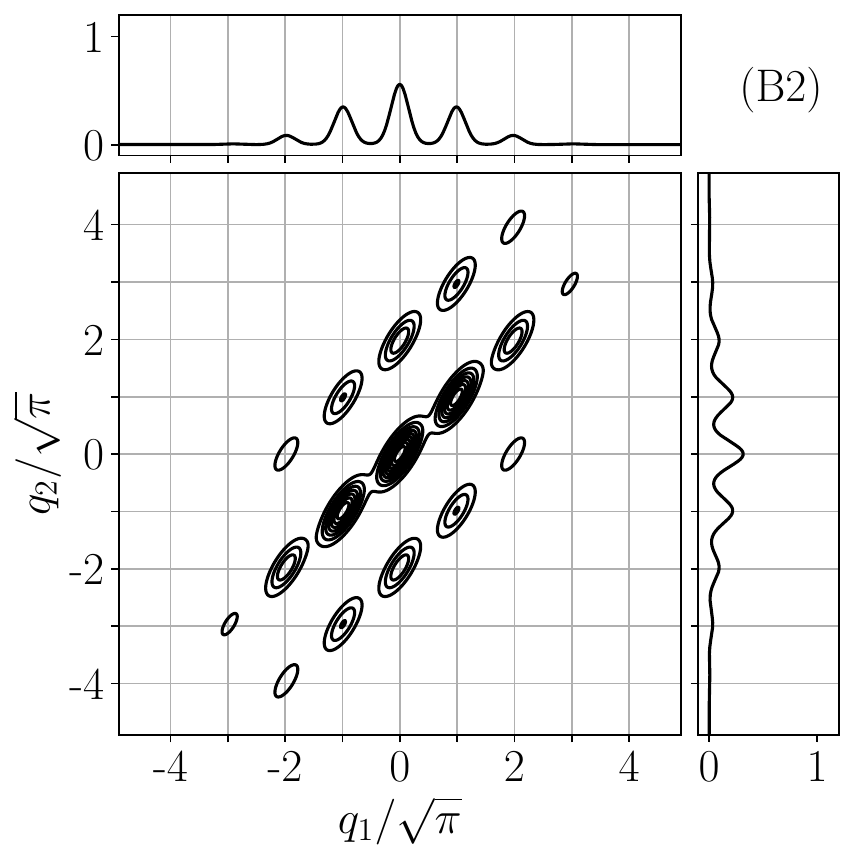}
    \caption{Two methods for preparing an approximation to the Bell state \( \ket{\Phi^+}_L \) represented by (A) and (B) corresponding to \cref{circ:qunaught bell} and \cref{circ:standard bell}, respectively. The approximation \( \ket{\argdot}_\Delta = e^{-\Delta \hat{n}} \ket{\argdot}_L \) is used in both cases. \textbf{(A1)}: The \( q \)-quadrature wavefunction \( \braket{q_1, q_2}{\qunaught, \qunaught}_\Delta \) of two qunaught states. \textbf{(A2)}: The result of applying a 50/50 beam splitter, transforming the state into \( \braket{q_1, q_2}{\Phi^+}_{\qunaught, \Delta} = \bra{q_1, q_2} \mathop{BS}_{12} \ket{\qunaught, \qunaught} \). Similarly, \textbf{(B1)} shows \( \braket{q_1, q_2}{+, 0}_\Delta \), and \textbf{(B2)} presents the result of applying a controlled-\(X\) gate, yielding \( \braket{q_1, q_2}{\Phi^+}_{0, \Delta} = \bra{q_1, q_2} \mathop{CX}_{12} \ket{+, 0}_\Delta \). Notably, the Bell state approximation obtained using qunaught states is symmetric between the two modes, with widths equal to those of the prepared input states. In contrast, the approximation obtained from the standard preparation circuit results in an asymmetrically widened second mode.}
    \label{fig:bell states}
\end{figure*}

\section{Notation and fundamental operations} \label{sec:notation}
This section introduces the essential notations and fundamental operations that form the foundation for the discussions and analyses in the subsequent sections. We define the hermitian quadrature operators $\hat{q}$ and $\hat{p}$, referred to as position and momentum respectively, with the canonical commutation relation $[\hat{q}, \hat{p}] = i$. Unless stated otherwise, kets without a subscript represent position eigenkets, while momentum eigenstates are denoted with a subscript $p$:
\begin{equation}
    \hat{q} \ket{s} = s \ket{s}
        ,\quad
    \hat{p} \ket{s}_p = s \ket{s}_p
        .
\end{equation}
Position and momentum eigenstates are related by the Fourier transform:
\begin{equation}
    \braket[p]{p}{q} = (2\pi)^{-\half} e^{-ipq}.
\end{equation}
The ideal infinite-energy GKP states are denoted as:
\begin{gather}
    \ket{0}_L = \sum_{n \in \mathbb{Z}} \ket{2n \sqrt{\pi}}
        ,\quad
    \ket{1}_L = \sum_{n \in \mathbb{Z}} \ket{(2n + 1) \sqrt{\pi}}
        .
\end{gather}
which represent the logical states in the computational basis and are eigenstates of the logical Pauli-Z operator (to be defined below). Their wavefunctions are shown in \cref{fig:gkp states}. Eigenstates in the complementary basis, corresponding to the Pauli-X operator, can be similarly defined. We also introduce the qunaught state, a squeezed logical state used in a special version of the Knill-type error correction scheme:
\begin{equation}
    \ket{\qunaught}_L = \sum_{n \in \mathbb{Z}} \ket{n \sqrt{2 \pi}}.
\end{equation}
This state in itself encodes no logical information due to the squeezed spacing; hence the name ``qunaught''. 

Realistic GKP states can be approximated by finite energy states in several ways and the equivalence between different approximations has been discussed in ref. \cite{matsuura_equivalence_2020}. Here, we denote a generic finite-energy approximation of a GKP logical state as $\ket{\argdot}_\Delta$, where $\Delta$ indicates finite width of the individual peaks. We defer specifying a particular choice of approximation until \cref{sec:error analysis}, as our results are independent of this choice until then.

We now introduce the operators central to this work. Phase-space displacement operators are defined as
\begin{gather}
    X(s) = e^{-i s \hat{p}}
        ,\quad
    Z(s) = e^{i s \hat{q}}
        , \text{ and} \nonumber \\
    D(\vb{s}) = X(s_1) Z(s_2)
\end{gather}         
where their actions on position eigenstates are given by 
\begin{gather}
\quad
    X(s) \ket{q} = \ket{q + s}
        ,\quad
    Z(s) \ket{q} = e^{isq} \ket{q}
        .
\end{gather}
In the GKP code space, the logical Pauli gates $X$ and $Z$ are implemented as $X(\sqrt{\pi})$ and $Z(\sqrt{\pi})$, respectively. The Hadamard gate corresponds to the Fourier operator $F$ which performs a $\pi/2$ rotation in phase space):
\begin{gather}
    F = R(\pi/2) = e^{i \frac{\pi}{4} (\hat{q}^2 + \hat{p}^2)}
        ,\quad
    R(\pi) = F^2.
\end{gather}
The Fourier operator acts on the eigenstates as
\begin{gather}
\quad
    F \ket{s} = \ket{s}_p
        \qq{and}
    F \ket{s}_p = \ket{-s}
        .
\end{gather}
Controlled phase-space displacement operators are defined as:
\begin{align}
    CX_{12} = e^{-i \hat{q}_1 \hat{p}_2},
        & \quad
    CZ_{12} = e^{i \hat{q}_1 \hat{q}_2}
\end{align}
with their action on two-mode position eigenstates given by:
\begin{align}
    CX_{12} \ket{q, q'} = \ket{q, q' + q},
        & \quad 
    CZ_{12} \ket{q, q'} = e^{i q q'} \ket{q, q'}
        .
\end{align}
These correspond to the logical controlled-$X$ and controlled-$Z$ gates, respectively. %Hence, for the above operators, we adopt the standard circuit notation of their corresponding logical gates. 
However, unlike their logical counterparts, the controlled displacement operators are not Hermitian $CX^\dagger \neq CX$, requiring explicit notation for adjoint operations in circuit diagrams. Finally, we introduce the balanced 50/50 beam splitter operator: 
\begin{gather}
    BS_{12} = e^{\frac{\pi}{4} (\hat{a}_1^\dagger \hat{a}_2 - \hat{a}_1 \hat{a}_2^\dagger)}
\end{gather}
where $\hat{a} = (\hat{q} + i\hat{p}) / \sqrt{2}$ is the annihilation operator of the harmonic oscillator.
Its action on position eigenstates is:
\begin{gather}
    BS_{12} \ket{q, q'} = \ket{\frac{1}{\sqrt{2}}(q + q'), \frac{1}{\sqrt{2}}(q - q')}.
\end{gather}
In circuit diagrams, a beam splitter is represented by  a downward arrow, while its adjoint, obtained by swapping inputs, is depicted with an upward arrow, as illustrated here
\begin{equation}
    BS_{12}
        =
    \begin{tikzpicture}[baseline=-1mm, scale=0.8]
        \draw[rounded corners=\bsrad, thick] (-0.8, 0.5) -- (-\bswidth, 0.5) -- (\bswidth, -0.5) -- (0.8, -0.5);
        \draw[rounded corners=\bsrad, thick] (-0.8, -0.5) -- (-\bswidth, -0.5) -- (\bswidth, 0.5) -- (0.8, 0.5);
        \filldraw[color=black, fill=white] (-0.4, -0.05) rectangle (0.4, 0.05);
        \draw[thin, arrows = {-Stealth[inset=1pt, length=3pt, angle'=45, round]}] (0,0.3) -- (0,-0.3);
    \end{tikzpicture}
        =
    \begin{tikzpicture}[baseline=-1mm, scale=0.8]
        \draw[rounded corners=\bsrad, thick] (-0.8, 0.5) -- (-\bswidth, 0.5) -- (\bswidth, -0.5) -- (0.8, -0.5);
        \draw[rounded corners=\bsrad, thick] (-0.8, -0.5) -- (-\bswidth, -0.5) -- (\bswidth, 0.5) -- (0.8, 0.5);
        \filldraw[color=black, fill=white] (-0.4, -0.05) rectangle (0.4, 0.05);
        \draw[thin, arrows = {-Stealth[inset=1pt, length=3pt, angle'=45, round]}] (0,-0.3) -- (0,0.3);
        \node at (0.15, 0.3) {$^\dagger$};
    \end{tikzpicture}
        =
    BS_{21}^\dagger
        .
\end{equation}

\section{Analysis of the action of GKP error correction circuits} \label{sec:circuit analyses}
In this section, we analyse the action of GKP error correction circuits, focussing on the equivalences between different implementations. Specifically, we show that the Knill approach, whether using a controlled displacement gate or a beam splitter as the entangling gate, is mathematically equivalent. 
%This equivalence is illustrated in \cref{fig:Knill circuits}, and therefore both implementations can collectively be referred to as the Knill approach. 
Furthermore, we show that the Steane approach to GKP error correction, shown on the left-hand side of \cref{fig:Steane}, is equivalent to a Knill-type error correction scheme when a specific Bell state is used as the ancilla (shown on the right-hand side of \cref{fig:Steane}). This result establishes that the Steane approach is a special case of the Knill approach, obviating the need for a separate analysis of its error correction performance.

We begin by analysing the error correction circuit depicted on the left-hand side of \cref{fig:Knill circuits}. Let its action be denoted by $K_1(\vb{s})$ where $\vb{s} = (s_1, s_2)^T$. The general result, derived in \cref{app:circuit calculations}, is:
\begin{gather}
    K_1(\vb{s}) 
        =
    (\bra{s_1} \tensor \bra[p]{s_2} \tensor \idty) CX^\dagger_{21} (\idty \tensor \ket{\Phi^+}_\Delta)
        \\ =
    \frac{1}{\sqrt{2\pi}} \int \odif{q', q''} \braket{q', q''}{\Phi^+}_\Delta \ket{q''}\bra{q'} D(-\vb{s})
        ,
\end{gather}
Similarly, we denote the action of the circuit on the right-hand side of \cref{fig:Knill circuits} as $K_2(\vb{s})$. As detailed in \cref{app:circuit calculations}, this yields 
\begin{gather}
    K_2(\vb{s}) 
        =
    \frac{1}{\sqrt{2}} (\bra{\frac{s_1}{\sqrt{2}}} \tensor \bra[p]{\frac{s_2}{\sqrt{2}}} \tensor \idty) BS_{21} (\idty \tensor \ket{\Phi^+}_\Delta)
        \\ \propto
    \frac{1}{\sqrt{2\pi}} \int \odif{q', q''} \braket{q', q''}{\Phi^+}_\Delta \ket{q''}\bra{q'} D(-\vb{s})
        .
\end{gather}
where the proportionality is up to an irrelevant global phase. The normalization factor $1/\sqrt{2}$ accounts for the projection onto the states $2^{-\frac{1}{4}} \ket{s_1/\sqrt{2}}$ and $2^{-\frac{1}{4}} \ket{s_2/\sqrt{2}}_p$, that are normalised with respect to $\vb{s} \in \mathbb{R}^2$.

\iffalse
{\color{red} If this is not clear then: \color{gray} The factor of $\frac{1}{\sqrt{2}}$ can alternatively be thought of as the proper transformation of the probability distribution for the measurement outcomes. To see this, notice that the probability of measuring an outcome $\vb{t} \in \mathbb{R}^2$ in the homodyne detectors is given by
\begin{equation}
    P(\vb{t}) = \bra{\psi} \hat{L}(\vb{t})^\dagger \hat{L}(\vb{t}) \ket{\psi} \odif[2]{\vb{t}}
\end{equation}
for some continuously parametrised measurement outcomes $\hat{L}$. But then for $\vb{t} = \frac{\vb{s}}{\sqrt{2}}$ with $\odif[2]{\vb{t}} = \frac{\odif[2]{\vb{s}}}{2}$, we find
\begin{equation}
    P(\vb{s})
        =
    \bra{\psi} \hat{L}(\vb{t})^\dagger \hat{L}(\vb{t}) \ket{\psi} \odif[2]{\vb{t}} 
        =
    \bra{\psi} \hat{L}(\vb{s}/\sqrt{2})^\dagger \hat{L}(\vb{s}/\sqrt{2}) \ket{\psi} \frac{\odif[2]{\vb{s}}}{2}
        =
    \bra{\psi} \left[ \frac{1}{\sqrt{2}} \hat{L}(\vb{s}/\sqrt{2}) \right]^\dagger \left[ \frac{1}{\sqrt{2}} \hat{L}(\vb{s}/\sqrt{2}) \right] \ket{\psi} \odif[2]{\vb{s}}
\end{equation}
from which we see that to describe the probability distribution of $\vb{s}$ additionally to evaluate the expression at $\vb{s}/\sqrt{2}$ one must account for the transformation of the differential by including an additional factor of $\frac{1}{\sqrt{2}}$.}
\fi

By direct comparison, the two circuits are found to have the same action:
\begin{equation}
    K \define K_1 \propto K_2,
\end{equation}
and in particular, we find that their common action takes the form
\begin{equation}
    K(\vb{s}) = \Pi_\Delta D(-\vb{s})
\end{equation}
where the projection operator is
\begin{equation}
    \Pi_\Delta \define \frac{1}{\sqrt{2\pi}} \int \odif{q', q''} \braket{q', q''}{\Phi^+}_\Delta \ket{q''}\bra{q'}.
    \label{eq:projection operator}
\end{equation}
The fact that GKP error correction factors into a component depending only on the measurement result and another dependent only on the ancillary states is a well established result \cite{baragiola_all-gaussian_2019}. Crucially, this result confirms the equivalence of the two implementations, and highlights that the projection operator takes the particular form given in \cref{eq:projection operator}.

Next, we analyse the similarity between the Steane approach and a specific Knill-type approach to GKP error correction, as depicted in \cref{fig:Steane}. Numerical simulations of both methods, conducted using Strawberry Fields \cite{killoran_strawberry_2019}, reveal that they produce identical outputs. This equivalence is visually demonstrated in \cref{fig:wigner}, where the Wigner functions of their outputs from both approaches are shown to overlap perfectly. Furthermore, quantitative comparison confirms this result, with the maximum absolute error between the two outputs being less than $10^{-8}$.

\begin{figure*}
    \centering
    \includegraphics[width=\textwidth]{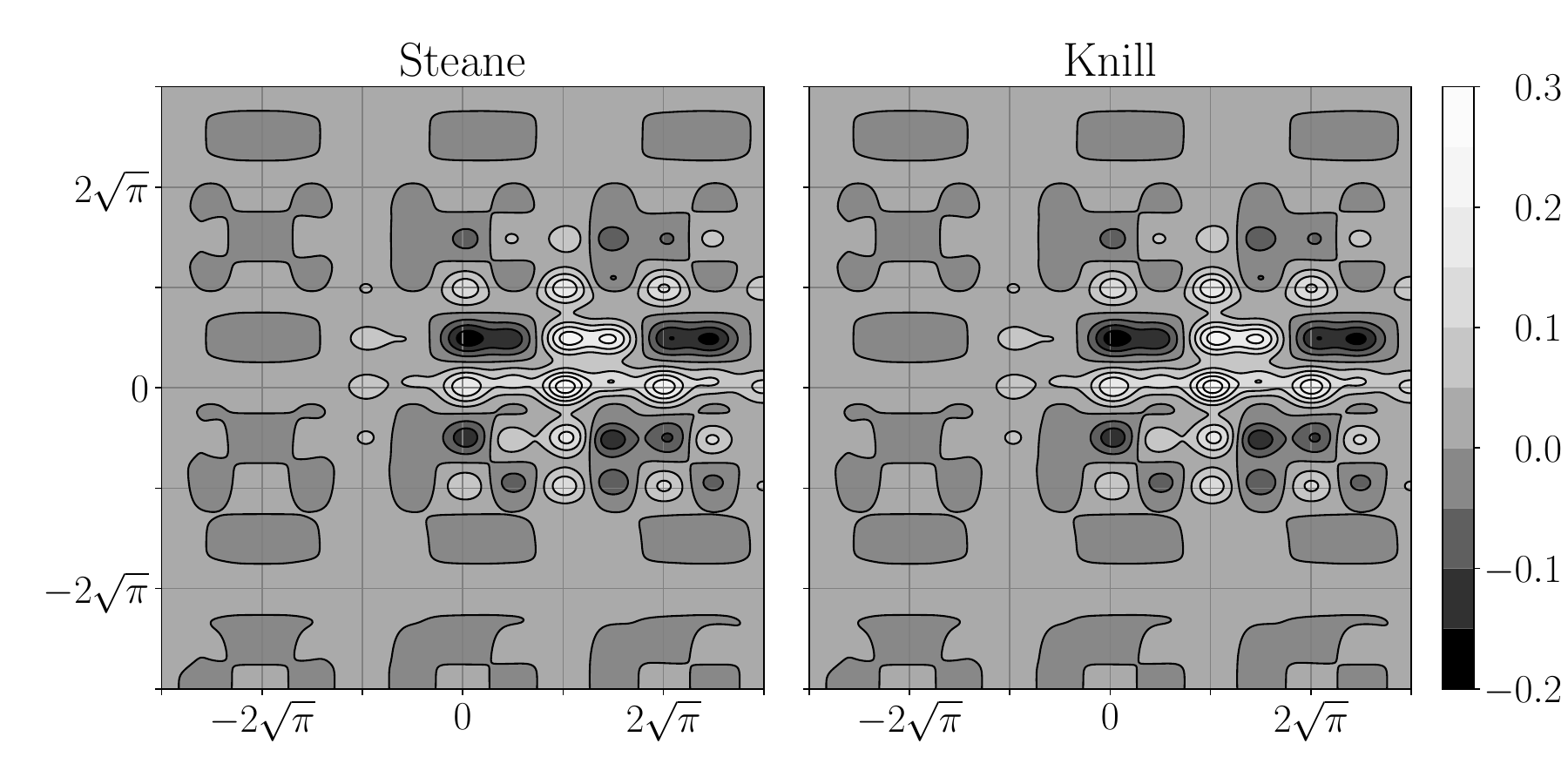}
    \caption{Wigner functions of the output states obtained from numerical simulations of the circuits in \cref{fig:Steane}, applied to the same coherent state with identical measurement outcomes. Simulations were performed using the Strawberry Fields Python package \cite{killoran_strawberry_2019}. Both visual inspection and numerical comparison (maximum error of \(10^{-8}\)) confirm that the two output states are identical.}
    \label{fig:wigner}
\end{figure*}

To prove this equivalence analytically, we consider the circuit on the left-hand side of \cref{fig:Steane}, whose action we denote by $S(\vb{s})$:
\begin{multline}
    S(\vb{s}) = (\hat{D}(-\vb{s}) \tensor \bra[p]{s_1} \tensor \bra[p]{s_2}) ~ \hat{F}_1 ~ CZ_{31} ~ \hat{F}_1^\dagger 
        \\
    \times ~ CZ_{21} ~ \hat{R}(\pi)_3 ~ (\idty \tensor \ket{0, 0}_\Delta)
\end{multline}
In \cref{app:circuit calculations}, we explicitly show that this circuit evaluates to
\begin{equation}
    S(\vb{s}) = \frac{1}{\sqrt{2\pi}} \int \odif{q, q''} \braket{q, q''}{\Phi^+}_{0, \Delta} \ket{q''}\bra{q} \hat{D}(-\vb{s})
\end{equation}
where $\ket{\Phi^+}_{0, \Delta}$ is a Bell state constructed using the following standard preparation circuit: 
\begin{circuit}
\begin{quantikz}[row sep={1cm, between origins}]
    \lstick{$\ket{0}_\Delta$} & \gate{F} & \ctrl{1} & \rstick[2]{$\ket{\Phi^+}_{0,\Delta}$} \\
    \lstick{$\ket{0}_\Delta$} &          & \targ{}  &
\end{quantikz}
\end{circuit}
Thus, the Steane approach is analytically shown to be equivalent to a specific case of the Knill approach, where the ancilla Bell state $\ket{\Phi^+}_{0, \Delta}$ is used: 
\begin{equation}
    \ket{\Phi^+}_\Delta = \ket{\Phi^+}_{0, \Delta}
        \implies
    S = K
        .
\end{equation}
This equivalence is exactly the relationship illustrated in \cref{fig:Steane}.

Finally, we note that the the projection operator $\Pi_{0, \Delta}$, implemented using the specific Bell state $\ket{\Phi^+}_{0, \Delta}$, can be expressed in terms of the Fourier transform $\fourier$ and convolution:
\begin{equation}
    \Pi_{0, \Delta}[\psi] = \frac{1}{\sqrt{2\pi}} \fourier^{-1}[\fourier[\psi] * \psi_{0, \Delta}] * \psi_{0, \Delta}.
\end{equation}
where $\psi(q) = \braket{q}{\psi}$ represents the wavefunction of the input state, and $\psi_{0, \Delta}(q) = \braket{q}{0}_\Delta$ corresponds to the wavefunction of the GKP ancilla state. This form directly reflects the structure of the left-hand circuit in \cref{fig:Steane} and suggests an efficient classical simulation strategy using the FFT algorithm.

\section{General analysis of the projection operator} \label{sec:analysis of pi}
We analyse the action of the GKP error correction projection operator, showing that if both the input state and the ancilla Bell state are well-approximated by linear superpositions of similar Gaussians, the resulting state after error correction is also a linear superposition of similar Gaussians. Furthermore, we derive the exact width of the spikes and their locations in phase space of the corrected state.

To analyse the projection operator $\Pi_\Delta$ (defined in \cref{eq:projection operator}), we work in the Schrödinger picture:
\begin{equation}
    \bra{q} \Pi_\Delta \ket{\psi} = \frac{1}{\sqrt{2\pi}} \int \odif{q'} \braket{q', q}{\Phi^+}_\Delta \braket{q'}{\psi}.
    \label{projected_wavefuntion}
\end{equation}
We assume the states are well-approximated by linear superpositions of similar Gaussian spikes which share the same covariance matrix. Under this assumption, we show that the error-corrected states remains a superposition of Gaussian spikes.

We represent a (non-normalised) Gaussian function as
\begin{equation}
    G(x; \mu, Q) \define e^{-\half (x - \mu)^T Q (x - \mu)}.
\end{equation}
Notice we use a formulation based on the precision matrix $Q$ rather than the covariance matrix $Q^{-1}$.
%(not to be confused with the covariance matrix between $q$ and $p$ of a Gaussian quantum state). 
We do not require the mean $\mu$ nor precision $Q$ to be real, allowing for phases that are linear or quadratic in the quadrature. The only constraint is that the covariance matrix must be symmetric positive semidefinite. Furthermore, the positions of the Gaussian spikes, represented by the means $\mu$, can belong to any countably discrete set. This set may for example correspond to a perfect GKP lattice, a slightly displaced lattice, or even an irregular lattice.

The input state $\ket{\psi}$ and the ancilla Bell state $\ket{\Phi^+}_\Delta$ are then approximated as
\begin{align}
    \braket{q}{\psi} \simeq \sum_{\mu_{in}} c_{\mu_{in}} G(q; \mu_{in}, Q_{in})
        &\qq{with}
    q \in \mathbb{R}
        \\
    \braket{\vb{q}}{\Phi^+}_\Delta \simeq \sum_{\vb{\mu}} c'_{\vb{\mu}} G(\vb{q}; \vb{\mu}, Q)
        &\qq{with}
    \vb{q} \in \mathbb{R}^2
        .
\end{align}
where $1/Q_{in}=\sigma_{in}^2$ represents the width of the spikes of the input state. We assume no relation between the spikes of the input $Q_{in} \in \mathbb{R}$ and those of the prepared ancilla Bell state $Q \in \mathbb{R}^{2\times2}$. Substituting these approximations into the expression \cref{projected_wavefuntion}, the wave function after error correction takes the form
\begin{multline}
    \bra{q} \Pi_\Delta \ket{\psi}
        =
    \frac{1}{\sqrt{2\pi}} \sum_{\vb{\mu}, \mu_{in}} c'_{\vb{\mu}} c_{\mu_{in}} 
        , \\ \times
    \int \odif{q'} G((q', q)^T; \vb{\mu}, Q) G(q'; \mu_{in}, Q_{in})
        . \label{eq:special integral}
\end{multline}
where the integral determines the interaction between the input state and the ancilla Bell state. The explicit computation of this integral is detailed in \cref{app:integral} \cref{lemma:special integral}, resulting in:
\begin{multline}
    \frac{1}{\sqrt{2\pi}} \int \odif{q'} G((q', q)^T; \vb{\mu}, Q) G(q'; \mu_{in}, Q_{in})
        \\ = 
    \frac{1}{N} G(\mu_1-\mu_{in}; 0, \rho^{-2}) ~ G(q; \mu_{out}^{\vb{\mu}, \mu_{in}}, \sigma_{out}^{-2})
        ,
\end{multline}
where $N = \sqrt{Q_{11} + Q_{in}}$ and $\rho^2 = \frac{1}{Q_{in}} + \frac{Q_{22}}{\det{Q}}$, and the key parameters are
\begin{equation}
    \mu_{out}^{\vb{\mu}, \mu_{in}} = \mu_2 + \frac{Q_{in} Q_{12}}{\det{Q} + Q_{in} Q_{22}} (\mu_1 - \mu_{in}),
        \label{eq:mean}
\end{equation}
and
\begin{equation}
    \sigma_{out}^2 = \frac{Q_{11} + Q_{in}}{\det{Q} + Q_{in} Q_{22}}.
    \label{eq:width}
\end{equation}
Thus, the resulting wave function is a superposition of Gaussian spikes with a common width $\sigma_{out}^2$:
\begin{equation}
    \bra{q} \Pi_\Delta \ket{\psi} = \sum_{\vb{\mu}, \mu_{in}} c''_{\vb{\mu}, \mu_{in}} G(q; \mu_{out}^{\vb{\mu}, \mu_{in}}, \sigma_{out}^{-2}).
        \label{eq:projected wave function}
\end{equation}
This shows that the post-correction error is fully characterised by the two parameters $\mu_{out}^{\vb{\mu}, \mu_{in}}$ and $\sigma_{out}^2$.

While the above analysis applies to the $q$ quadrature wave function, a similar analysis can be performed for the $p$ quadrature. In this case
% \begin{gather}
%     \bra[p]{p} \Pi_\Delta \ket{\psi}
%         =
%     \frac{1}{\sqrt{2\pi}} \int \odif{q', q''} \braket{q', q''}{\Phi^+}_\Delta \braket[p]{p}{q''} \braket{q'}{\psi}
%         \\ =
%     \frac{1}{\sqrt{2\pi}} \int \odif{q', q''} \braket{q', q''}{\Phi^+}_\Delta \braket[p]{p}{q''} \int \odif{p'} \braket{q'}{p'}_p \braket[p]{p'}{\psi}
%         \\ =
%     \frac{1}{\sqrt{2\pi}} \int \odif{q', q'', p'} \braket[p]{-p', p}{q', q''} \braket{q', q''}{\Phi^+}_\Delta \braket[p]{p'}{\psi}
%         \\ =
%     \frac{1}{\sqrt{2\pi}} \int \odif{p'} \braket[p]{-p', p}{\Phi^+}_\Delta \braket[p]{p'}{\psi}
% \end{gather}
\begin{equation}
    \bra[p]{p} \Pi_\Delta \ket{\psi}
        =
    \frac{1}{\sqrt{2\pi}} \int \odif{p'} \braket[p]{p', p}{\Phi^+}_\Delta \braket[p]{-p'}{\psi}
        ,
\end{equation}
This expression mirrors the q-quadrature analysis with the substitution 
%which, up to the sign change, is exactly the same form as has been studied above, but where everything is expressed in the $p$ quadrature. It is then evident that the sign change is equivalent to the transformation 
$\mu_{in} \mapsto -\mu_{in}$ in the result in \cref{eq:projected wave function}. This transformation does not affect the overall conclusion, allowing us to use the same results interchangeably for both quadratures in further analysis.

\section{Error analysis} \label{sec:error analysis}
\iffalse
In this work, we use the width $\Delta$ of the individual spikes in the wave function as a measure of quality, rather than the GKP squeezing $s$. However, following \cite{noh_fault-tolerant_2020, menicucci_fault-tolerant_2014}, these are directly related by:
\begin{equation}
    s = -10 \log_{10} \left( \frac{\Delta^2/2}{1/2} \right) \text{dB} = -10 \log_{10}(\Delta^2) ~ \text{dB},
\end{equation}
where $\Delta^2/2$ and $1/2$ represent the quadrature variances of the spike and the vacuum, respectively. (The factor of $1/2$ arises because probabilities are the square of the wave function.) For this reason, we will use the two terms interchangeably, noting that a decrease in spike width corresponds to an increase in squeezing.

As shown earlier, there is no fundamental difference between the Knill and Steane approaches; the distinction lies only in the preparation in the Bell state $\ket{\Phi^+}_\Delta$. For a Bell state consisting of Gaussian spikes described by the covariance matrix $Q^{-1}$, the output after error-correction will have spikes with widths:
\begin{equation}
    \sigma_{out}^2 = \frac{1}{Q_{22}} ~ \frac{Q_{11} + Q_{in}}{\frac{\det{Q}}{Q_{22}} + Q_{in}},
        \label{eq:width}
\end{equation}
where $Q_{in}^{-1} = \sigma_{in}^2$ represents the width of the Gaussian spikes of the input. Therefore, determining the output quality reduces to analysing the covariance matrix $Q^{-1}$ for the Bell states prepared by the two preparation schemes in \cref{circ:standard bell,circ:qunaught bell}.
\fi

Having derived the post-error-corrected state in \cref{eq:projected wave function} with the GKP squeezing width defined in \cref{eq:width}, we now explicitly quantify the performance of the Knill-type error correction scheme using two types of entangled Bell states as ancilla: those produced through the qunaught-based circuit and those generated via the traditional qubit-based circuit.

We consider an input state that has undergone errors, mathematically described by the photon number dampening operator, $e^{-\Delta^2\hat{N}}$, resulting in the state\cite{menicucci_fault-tolerant_2014}:
\begin{gather}
    \ket{\psi}_\Delta = e^{-\Delta^2 \hat{N}} \ket{\psi}_L
\end{gather}
where $\Delta \in \mathbb{R}$ and $\hat{N} = \sum_i \hat{n}_i$ is the total photon number operator. This noise model is symmetric, producing approximate GKP states with equal squeezing in both the $q$ and $p$ quadratures. Under this approximation, the width of the Gaussian spikes is well-approximated by $Q \simeq \Delta^{-2} \idty$ as shown in  \cite{noh_fault-tolerant_2020}. This symmetry simplifies analysis and makes the model particularly useful for studying error correction performance. While the photon number dampening approximation is used here, the result in \cref{eq:width} is general and can be readily applied to other noise models. Considering the equivalence between the three main GKP noise models established in \cite{matsuura_equivalence_2020}, the conclusions drawn here are broadly applicable across all these cases.

\subsection{Bell state from qunaught states} \label{sec:error qunaught}
Consider the preparation scheme where two qunaught states are entangled using a beam splitter: 
\begin{circuit}
\begin{quantikz}[row sep={1cm, between origins}]
    \lstick{$\ket{\qunaught}_\Delta$} & \bsdown{} & \rstick[2]{$\ket{\Phi^+}_{\qunaught, \Delta}$} \\
    \lstick{$\ket{\qunaught}_\Delta$} &           &
\end{quantikz}
\label{circ:qunaught bell}
\end{circuit}
The beam splitter is a passive linear transformation that preserves the total number of photons and therefore commutes with the number operator $\hat{N}$. Therefore, the circuit can be rewritten as
\begin{align}
    \ket{\Phi^+}_{\qunaught, \Delta}
        &=
    BS_{12} \ket{\qunaught, \qunaught}_\Delta
        =
    BS_{12} ~ e^{-\Delta^2 \hat{N}} \ket{\qunaught, \qunaught}_L
        \\ &=
    e^{-\Delta^2 \hat{N}} BS_{12} \ket{\qunaught, \qunaught}_L
        =
    e^{-\Delta^2 \hat{N}} \ket{\Phi^+}_L
        .
\end{align}
The result shows that the approximate Bell state prepared by the circuit is equivalent to applying the symmetric approximation directly to the ideal logical Bell state. Consequently, the Bell state comprises Gaussian spikes described by $Q = \Delta^{-2} \idty$, as confirmed by simulations in \cref{fig:bell states} (A1) and (A2).

Since $Q$ is diagonal, we can immediately see from \cref{eq:width} that
\begin{equation}
    \sigma_{out}^2 = \Delta^2.
\end{equation}
This means the output width matches that of the prepared ancilla Bell state, regardless of the input width. Since the error model and the results are symmetric between the $q$ and $p$ quadratures, the same applies to both quadratures.

Finally, since $Q$ is diagonal, $Q_{12}=0$, and thus the corrected Gaussian spikes are located exactly at the locations defined by the prepared ancilla Bell state. In other words, the displacement errors post error correction are independent of any previous displacement errors. If the Bell state aligns with the GKP lattice, then the corrected state will be free of displacement errors.

In summary, the qunaught preparation scheme for preparing the approximate Bell states for the Knill approach produces error-corrected states with symmetric $q$ and $p$ quadratures, preserving the GKP squeezing of the qunaught states and avoiding displacements.

\subsection{Bell state from standard qubit circuit} \label{sec:error standard}
Now consider the standard qubit circuit for preparing the Bell state:
\begin{circuit}
\begin{quantikz}[row sep={1cm, between origins}]
    \lstick{$\ket{0}_\Delta$} & \gate{F} & \ctrl{1} & \rstick[2]{$\ket{\Phi^+}_{0, \Delta}$} \\
    \lstick{$\ket{0}_\Delta$} &          & \targ{}  &
\end{quantikz}
\label{circ:standard bell}
\end{circuit}
Following similar reasoning, the Bell state can be expressed as
\begin{align}
    \ket{\Phi^+}_{0, \Delta}
        &=
    CX_{12} \hat{F}_1 \ket{0, 0}_\Delta
        \\ &=
    CX_{12} \hat{F}_1 ~ e^{-\Delta^2 \hat{N}} \ket{0, 0}_L
        \\ &=
    CX_{12} ~ e^{-\Delta^2 \hat{N}} \ket{+, 0}_L
        .
\end{align}
Unlike the beam splitter, the controlled-X ($CX$)-gate does not commute with the number operator $\hat{N}$. Instead, its action transforms the quadratures according to the symplectic matrix: 
%we have $CX_{12} \ket{\vb{q}} = \ket{q_1, q_1+q_2} = \ket{S \vb{q}}$, where $S$ is (the $q$-quadrature part of) the symplectic form of the $CX$-gate. 
\begin{equation}
    S = \mat{1&0\\1&1}
        \qq{and}
    S^{-1} = \mat{1&0\\-1&1}
        .
\end{equation}
This transformation changes the wave function as $\bra{\vb{q}} CX_{12} \ket{\psi} = \braket{S^{-1} \vb{q}}{\psi}$, and modifies the the mean and precision matrix of Gaussian states as $G(S^{-1} \vb{q}; \vb{\mu}, Q) = G(\vb{q}; S \vb{\mu}, S^{-T} Q S^{-1})$. 

The approximate logical state $e^{-\Delta^2 \hat{N}} \ket{+, 0}_L$ has spikes with inverse covariance matrix $\Delta^{-2} \idty$, and so, the approximate Bell state $\ket{\Phi^+}_{0, \Delta}$ is obtained from this by transformation according to $S$:
\begin{equation}
    Q = S^{-T} (\Delta^{-2} \idty) S^{-1}
        =
    \Delta^{-2} \mat{2&-1\\-1&1}
        .
\end{equation}
The diagonal entries $Q_{11}$ and $Q_{22}$ differ, revealing an asymmetry of the prepared Bell state, as illustrated in \cref{fig:bell states}.  

Using \cref{eq:width}, the output spike width in the $q$-quadrature is
\begin{equation}
    \sigma_{out, q}^2 
        =
    \Delta^2 ~ \frac{1 + 2 (\frac{\sigma_{in, q}}{\Delta})^2}{1 + (\frac{\sigma_{in, q}}{\Delta})^2}
        \define
    \Delta^2 \cdot e^{2 r_q}
        .
        \label{q-width}
\end{equation}
This result shows that the output width depends on the input width relative to the prepared ancilla states. Specifically, $\sigma_{out, q}^2 > \Delta^2$, stating that the output width is always larger than that of the prepared ancilla $\ket{0}_\Delta$-states. This indicates that the $q$-quadrature GKP squeezing noise of the signal is not corrected to a width of $\Delta^2$ as in the Knill type error correction scheme but to an amplified value. For $\sigma_{in}>>\Delta$ we find that $\sigma_{out, q}^2=2\Delta^2$.  

Repeating the analysis for the $p$-quadrature, where the $CX_{12}$ gate transforms momentum eigenstates as $\bra[p]{\vb{p}} CX_{12} \ket{\psi} = \braket[p]{S^T \vb{p}}{\psi}$, the precision matrix becomes:
\begin{equation}
    Q = 
    \Delta^{-2} \mat{1&1\\1&2}
        .
\end{equation}
We notice that compared to the previous case the role of the diagonal entries $Q_{11}$ and $Q_{22}$ has been swapped. Since the off-diagonal entries contribute to the output width only through the determinant, their sign changes cancel, and the above covariance is equivalent to simply swapping the two modes of the Bell state. Using \cref{eq:width}, the output spike width in the $p$-quadrature is  
\begin{equation}
    \sigma_{out, p}^2
        =
    \Delta^2 ~ \frac{1 + (\frac{\sigma_{in, p}}{\Delta})^2}{2 + (\frac{\sigma_{in, p}}{\Delta})^2}
    \define
    \Delta^2 \cdot e^{2 r_p}
        .
        \label{p-width}
\end{equation}
The result still depends on the GKP squeezing of the input state compared to that of the prepared ancilla states, this time, however, we have $\sigma_{out, p}^2 < \Delta^2$, and so, the resulting width is always smaller than that of the prepared ancilla $\ket{0}_\Delta$-states. This suggests a noise correction of the $p$-quadrature beyond the GKP squeezing variance of the ancilla state.

We have now found the spike width in both the $q$ and $p$ quadratures post error correction, given in \cref{q-width} and \cref{p-width}. Firstly, we have already observed that the error correction is indeed asymmetric between the two quadratures. Secondly, we notice that the above equations are very similar to the relations for the application of ideal squeezing: $\sigma_q^2 \mapsto e^{2r} \sigma_q^2$ and $\sigma_p^2 \mapsto e^{-2r}\sigma_p^2$. An ideal squeezing operation does not change the total amount of GKP squeezing, but only redistributes what is already present between the two quadratures. This follows from the fact that squeezing by an amount $r$ preserves the Heisenberg product $\sigma_q^2 ~ \sigma_p^2 \mapsto e^{2r} \sigma_q^2 ~ e^{-2r}\sigma_p^2 = \sigma_q^2 ~ \sigma_p^2$. If one quadrature is squeezed the other is simultaneously anti-squeezed by the same amount.

There is no reason for us to expect GKP error correction to preserve total squeezing between $q$ and $p$, and indeed, it turns out not to be the case. To see this, we take inspiration from ideal squeezing, and consider the Heisenberg product for the output state. To ease notation, we also write $x = \frac{\sigma_{in, q}}{\Delta}$ and $y = \frac{\sigma_{in, p}}{\Delta}$.
\begin{equation}
    e^{2 r_q} e^{2 r_p} 
        =
    \frac{1 + 2 x^2}{1 + x^2} \cdot \frac{1 + y^2}{2 + y^2}
        =
    \frac{x^2 y^2 + z^2}{1 + z^2}
\end{equation}
where $z^2 = 1 + 2 x^2 + y^2 + x^2 y^2$. The right-most fraction is easily seen to be $\geq 1$, corresponding to a gain of width, or equivalently a net loss in total GKP squeezing exactly when $x^2 y^2 \geq 1$. This is equivalent to $\sigma_{in, q}^2 \sigma_{in, p}^2 \geq \Delta^2 \Delta^2$, that is, when the combined quadrature width of the input is greater than that of the prepared states, which will certainly be the case for most real cases.

As in the previous case, we also consider what happens to the locations of the Gaussian spikes of the error-corrected state. As we have shown, the inverse covariance matrix of the spikes in the prepared Bell state is \emph{not} diagonal. By \cref{eq:mean}, we immediately see that this implies that the locations of the corrected spikes will depend on the locations in the noisy input state. In other words, we expect to see (small) displacements away from the GKP lattice of the peaks in the error corrected state; displacement errors that are correlated with previous errors.

There is one final subtlety that must be handled. In the above, we have studied the case where in the Knill approach, the input is entangled with the mode that was initially the control for the controlled displacement operation in the preparation circuit. The preparation circuit is not symmetric between the two modes, which means that we should also study the case where we swap the two modes. However, as we have seen, the difference between considering the $q$ and $p$ quadratures of the output corresponds to swapping the modes of the Bell state. Thus, this final case is actually covered by the above simply by swapping the results for the two quadratures.

To conclude, we have shown that the Knill approach using the standard qubit circuit for preparing an approximate Bell state (which is the same as using the Steane approach), produce error corrected states that are asymmetric between $q$ and $p$. Furthermore, in most cases, the total squeezing of the error corrected state is less than that of the prepared approximate logical states. That is, some of the work done preparing high quality states is lost. Additionally, the corrected states will experience displacement errors off of those defined by the prepared Bell state.
\section{Conclusion}
In conclusion, our analysis definitively shows that the Steane approach to GKP error correction is equivalent to the Knill approach when using the standard Bell state preparation circuit (\cref{sec:circuit analyses}). However, we find that using the Bell state prepared using qunaught states results in error correction that yields higher quality logical states for the same amount of initial GKP squeezing, making it the preferable option (\cref{sec:error analysis}). Additionally, we have shown that the choice of entangling gate in the Knill approach is somewhat flexible: both the beam splitter and the controlled displacement operator result in equivalent physical error correction (\cref{sec:circuit analyses}). By specifically using a beam splitter in conjunction with Bell state preparation via qunaught states, GKP error correction can be achieved using only linear optics, thus simplifying the implementation. In this way, the qunaught-based approach provides both better error correction and ease of implementation without trade-off.

% From the analysis in Section 6.1, we have reason to believe that the Bell state prepared via qunaught states may indeed be optimal in the following sense. If the Bell state is prepared from a product state using only Gaussian operations, the Knill approach cannot produce error-corrected states with total GKP squeezing greater than that of the individual states of the initial product state. However, a formal proof of this remains elusive and is left for future work.

\section{Acknowledgements}
This work was supported by Innovation Fund Denmark under grant no. 1063-00046B - “PhotoQ Photonic Quantum Computing”. ULA also acknowledge support from the Danish National Research Foundation, Center for Macroscopic Quantum States (bigQ, DNRF0142), EU project CLUSTEC (grant agreement no. 101080173), EU ERC project ClusterQ (grant agreement no. 101055224) and the NNF project CBQS (NNF 24SA0088433).

Also, we thank in particular our colleagues Anton Alnor Christensen (Aarhus University), Sebastian Yde Madsen (Kvantify ApS), Josefine Bjørndal Robl (Aarhus University), and Caterina Vigliar (Technical University of Denmark) for their helpful discussions and detailed feedback on the manuscript.

\section{Additional resources}
The code used to make \cref{fig:gkp states,fig:bell states,fig:wigner} as well as code used for numerical sanity-checks of analytical results is publicly available at \cite{Marqversen_quantum_computations}.

\emergencystretch=1em
\printbibliography

\end{multicols}

\appendix
\appendixpage
\numberwithin{equation}{chapter}
\newcommand{\stepparentcounter}{\stepcounter{parentequation} \gdef\theparentequation{\Alph{chapter}.\arabic{parentequation}}}

\chapter{Explicit calculations for GKP error correction circuits} \label{app:circuit calculations}
We explicitly calculate the action of the different types of GKP error correction $K_1$, $K_2$, and $S$, as introduced in \cref{sec:circuit analyses} of the main text. First consider $K_1$ defined as the circuit in the left-hand side of \cref{fig:Knill circuits} of the main text.
\begin{subequations}
\begin{align}
    K_1(\vb{s}) 
        &= 
    (\bra{s_1} \tensor \bra[p]{s_2} \tensor \idty) ~ CX^\dagger_{21} ~ (\idty \tensor \ket{\Phi^+}_\Delta)
        \\ &=
    (\bra{s_1} \tensor \bra[p]{s_2} \tensor \idty) CX^\dagger_{21} \int \odif{q, q', q''} \ket{q, q', q''} \braket{q', q''}{\Phi^+}_\Delta \bra{q}
        \\ &=
    (\bra{s_1} \tensor \bra[p]{s_2} \tensor \idty) \int \odif{q, q', q''} \ket{q-q', q', q''} \braket{q', q''}{\Phi^+}_\Delta \bra{q}
        \\ &= 
    \frac{1}{\sqrt{2\pi}} \int \odif{q, q', q''} \delta(q-q' - s_1) e^{-i s_2 q'} \braket{q', q''}{\Phi^+}_\Delta \ket{q''}\bra{q}
        \\ &=
    \frac{1}{\sqrt{2\pi}} \int \odif{q', q''} \braket{q', q''}{\Phi^+}_\Delta \ket{q''}\bra{q' + s_1} e^{-i s_2 q'}
        \\ &=
    \frac{1}{\sqrt{2\pi}} \int \odif{q', q''} \braket{q', q''}{\Phi^+}_\Delta \ket{q''}\bra{q'} D(-\vb{s})
        ,
\intertext{Now consider $K_2$ defined as the circuit in the right-hand side of \cref{fig:Knill circuits} of the main text. In the definition of $K_2$, we include an additional factor of $1/\sqrt{2}$ which comes from projecting onto the states $2^{-\frac{1}{4}} \ket{s_1/\sqrt{2}}$ and $2^{-\frac{1}{4}} \ket{s_2/\sqrt{2}}_p$ that are normalised with respect to $\vb{s} \in \mathbb{R}^2$.} \stepparentcounter
    K_2(\vb{s}) 
        &=
    \frac{1}{\sqrt{2}} (\bra{\frac{s_1}{\sqrt{2}}} \tensor \bra[p]{\frac{s_2}{\sqrt{2}}} \tensor \idty) ~ BS_{21} ~ (\idty \tensor \ket{\Phi^+}_\Delta)
        \\ &=
    \frac{1}{\sqrt{2}} (\bra{\frac{s_1}{\sqrt{2}}} \tensor \bra[p]{\frac{s_2}{\sqrt{2}}} \tensor \idty) BS_{21} \int \odif{q, q', q''} \ket{q, q', q''} \braket{q', q''}{\Phi^+}_\Delta \bra{q}
        \\ &=
    \frac{1}{\sqrt{2}} (\bra{\frac{s_1}{\sqrt{2}}} \tensor \bra[p]{\frac{s_2}{\sqrt{2}}} \tensor \idty) \int \odif{q, q', q''} \ket{\frac{1}{\sqrt{2}} (q-q'), \frac{1}{\sqrt{2}} (q+q'), q''} \braket{q', q''}{\Phi^+}_\Delta \bra{q}
        \\ &=
    \frac{1}{\sqrt{2\pi}} \int \odif{q, q', q''} \frac{1}{\sqrt{2}} \delta(\frac{1}{\sqrt{2}}(q-q' - s_1)) e^{-i \half s_2 (q+q')} \braket{q', q''}{\Phi^+}_\Delta \ket{q''}\bra{q}
        \\ &=
    \frac{1}{\sqrt{2\pi}} \int \odif{q', q''} \braket{q', q''}{\Phi^+}_\Delta \ket{q''}\bra{q' + s_1} e^{-i s_2 q'} e^{-i \half s_1 s_2}
        \\ &=
    \frac{1}{\sqrt{2\pi}} \int \odif{q', q''} \braket{q', q''}{\Phi^+}_\Delta \ket{q''}\bra{q'} D(-\vb{s}) e^{-i \half s_1 s_2}
        .
\intertext{Finally consider $S$ defined as the circuit in the left hand side of \cref{fig:Steane} of the main text. For this computation, we utilise the following identity: $F_1 CZ_{31} F_1^\dagger = CX_{31}^\dagger$.} \stepparentcounter
    S(\vb{s}) &= (\hat{D}(-\vb{s}) \tensor \bra[p]{s_1} \tensor \bra[p]{s_2}) ~ \hat{F}_1 ~ CZ_{31} ~ \hat{F}_1^\dagger ~ CZ_{21} ~ \hat{R}(\pi)_3 ~ (\idty \tensor \ket{0, 0}_\Delta)
        \\ &=
    (\hat{D}(-\vb{s}) \tensor \bra[p]{s_1} \tensor \bra[p]{s_2}) ~ CX^\dagger_{31} ~ CZ_{21} ~ \hat{R}(\pi)_3 \int \odif{q, q', q''} \ket{q, q', q''} \braket{q', q''}{0, 0}_\Delta \bra{q}
        \\ &=
    (\hat{D}(-\vb{s}) \tensor \bra[p]{s_1} \tensor \bra[p]{s_2}) ~ CX^\dagger_{31} \int \odif{q, q', q''} e^{iqq'} \ket{q, q', -q''} \braket{q', q''}{0, 0}_\Delta \bra{q}
        \\ &=
    (\hat{D}(-\vb{s}) \tensor \bra[p]{s_1} \tensor \bra[p]{s_2}) \int \odif{q, q', q''} e^{iqq'} \ket{q + q'', q', -q''} \braket{q', q''}{0, 0}_\Delta \bra{q}
        \\ &=
    \hat{D}(-\vb{s}) \int \odif{q, q', q''} e^{iqq'} \frac{1}{2\pi} e^{-i s_1 q'} e^{i s_2 q''} \ket{q + q''} \braket{q', q''}{0, 0}_\Delta \bra{q}
        \\ &=
    \frac{1}{2\pi} \int \odif{q, q', q''} e^{i (q-s_1) q'} e^{-i s_2 q} \ket{q - s_1 + q''} \braket{q', q''}{0, 0}_\Delta \bra{q}
        \\ &=
    \frac{1}{2\pi} \int \odif{q, q', q''} e^{i q q'} \braket{q', q''}{0, 0}_\Delta \ket{q + q''} \bra{q + s_1} e^{-i s_2 q}
        \\ &=
    \frac{1}{\sqrt{2\pi}} \int \odif{q, q''} \left[ \frac{1}{\sqrt{2\pi}} \int \odif{q'} e^{i q q'} \braket{q', q'' - q}{0, 0}_\Delta \right] \ket{q''}\bra{q} \hat{D}(-\vb{s})
        .
\end{align}
\end{subequations}
Comparing this result to those found from analysing the Knill approach circuits, we are motivated to further consider the term in brackets:
\begin{align}
    \frac{1}{\sqrt{2\pi}} \int \odif{q'} e^{i q q'} \braket{q', q'' - q}{0, 0}_\Delta &
        =
    \bra{q, q'' - q} \hat{F}_1 \ket{0, 0}_\Delta
        =
    \bra{q, q''} CX_{21} \hat{F}_1 \ket{0, 0}_\Delta
        =
    \braket{q, q''}{\Phi^+}_{0, \Delta}
        ,
\end{align}
where $\ket{\Phi^+}_{0, \Delta}$ is exactly the Bell state defined in \cref{circ:standard bell} of the main text. Thus
\begin{equation}
    S(\vb{s}) = \frac{1}{\sqrt{2\pi}} \int \odif{q, q''} \braket{q, q''}{\Phi^+}_{0, \Delta} \ket{q''}\bra{q} \hat{D}(-\vb{s}).
\end{equation}

\chapter{Special integral of product of Gaussians} \label{app:integral}
Using the same notation as in the main text
\begin{equation}
    G(x; \mu, Q) = e^{-\half (x - \mu)^T Q (x - \mu)}
\end{equation}
we will start by showing a general identity for the product of multivariate Gaussians in \cref{lemma:product of gaussians}. This result is then directly applied in the proof of \cref{lemma:special integral} which is the main calculation referenced in \cref{sec:analysis of pi} of the main text.

\begin{lemma} \label{lemma:product of gaussians}
    Let $\delta \in \mathbb{R}^N$, $A, B \in \mathbb{R}^{N \times N}$ be symmetric matrices and assume that their sum $V = A + B$ is invertible ($A$ and $B$ do not have to be individually invertible). Then
    \begin{gather}
        G(x; \delta, A) G(x; 0, B) = G(0; \delta, D) G(x; \alpha, V)
            \intertext{where}
        V = A + B
            ,\hspace{2mm}
        D = A V^{-1} B
            ,\hspace{2mm} \text{and} \hspace{2mm}
        \alpha = V^{-1} A \delta
            .
    \end{gather}
\end{lemma}
\begin{proof}
    Consider the combined exponent of the product, and complete the square (using that $A$ and $V$ are symmetric and $V$ invertible):
    \begin{gather}
        (x - \delta)^T A (x - \delta) + x^T B x
            =
        x^T (A + B) x - 2 x^T A \delta + \delta^T A \delta
            =
        x^T V x - 2 x^T V (V^{-1} A \delta) + \delta^T A \delta
            \\ =
        x^T V x - 2 x^T V \alpha + \delta^T A \delta
            =
        (x - \alpha)^T V (x - \alpha) - \alpha^T V \alpha + \delta^T A \delta
            .
    \end{gather}
    Consider further the second and third terms
    \begin{gather}
        - \alpha^T V \alpha + \delta^T A \delta 
            =
        - \delta^T A^T V^{-T} V V^{-1} A \delta + \delta^T A \delta
            =
        \delta^T (A - A V^{-1} A) \delta
            =
        \delta^T A V^{-1} (V - A) \delta
            \\ =
        \delta^T A V^{-1} B \delta
            =
        \delta^T D \delta
            .
    \end{gather}
    Altogether we find that the combined exponent exactly match that which we are after:
    \begin{equation}
        (x - \delta)^T A (x - \delta) + x^T B x
            =
        (x - \alpha)^T V (x - \alpha) + \delta^T D \delta
            .
    \end{equation}
\end{proof}

\begin{lemma}[Special integral of product of Gaussians] \label{lemma:special integral}
    Let $\mu \in \mathbb{R}$, $\vb{\mu} \in \mathbb{R}^2$ and $q \in \mathbb{R}_{>0}$, $Q \in \mathbb{R}^{2 \times 2}$ where $Q$ is symmetric and positive definite. Then the following formula holds:
    \begin{equation}
        \frac{1}{\sqrt{2\pi}} \int \odif{x_1} G(\vb{x}; \vb{\mu}, Q) G(x_1; \mu, q)
            =
        \frac{1}{N} G(0; \mu_1-\mu, \rho^{-2}) G(x_2; \nu, \sigma^{-2})
    \end{equation}
    with the following definitions
    \begin{align}
        N = \sqrt{Q_{11} + q}
            ,\quad
        \rho^2 = \frac{1}{q} + \frac{Q_{22}}{\det{Q}}
            ,\quad
        \nu = \mu_2 + \frac{q Q_{12}}{\det{Q} + q Q_{22}} (\mu_1 - \mu)
            ,\qq{and}
        \sigma^2= \frac{Q_{11} + q}{\det{Q} + q Q_{22}}
            .
    \end{align}
\end{lemma}

\begin{proof}
    To ease notation, we start by labelling the integral by $I$. By a simple translational transformation we find:
    \begin{equation}
        I = \frac{1}{\sqrt{2\pi}} \int \odif{x_1} G(\vb{x}; \vb{\mu}, Q) G(x_1; \mu, q)
            =
        \frac{1}{\sqrt{2\pi}} \int \odif{x_1} G(\vb{x}; \vb{\delta}, Q) G(x_1; 0, q)
    \end{equation}
    where we have introduced variable $\vb{\delta} = (\mu_1 - \mu, \mu_2)^T$. Now introduce the following definitions
    \begin{equation}
        Q' = \mat{q & 0 \\ 0 & 0}
            ,\quad
        V = Q + Q'
            ,\quad
        D = Q V^{-1} Q'
            ,\qq{and}
        \vb{\alpha} = V^{-1} Q \vb{\delta}
            .
    \end{equation}
    With these we invoke \cref{lemma:product of gaussians} which gives
    \begin{equation}
        I = G(\vb{0}; \vb{\delta}, D) \frac{1}{\sqrt{2\pi}} \int \odif{x_1} G(\vb{x}; \vb{\alpha}, V).
    \end{equation}
    Notice that $V$ is indeed invertible which is necessary to invoke the lemma. This is seen by using the fact that $Q$ is positive definite and $q>0$, from which it follows that $\det{V} = \det{Q} + q Q_{22} > 0$. Now use the fact that the marginal distribution of a multivariate normal distribution is simply the normal distribution obtained by deleting all rows and columns corresponding to the integrated out variable. We get
    \begin{equation}
        I = \sqrt{\frac{\det(V^{-1})}{[V^{-1}]_{22}}} G(\vb{0}; \vb{\delta}, D) G(x_2; \alpha_2, [V^{-1}]_{22}^{-1}).
        \label{eq:obfuscated result}
    \end{equation}
    We now turn to compute the relevant values in this expression explicitly. Doing this, we find
    \begin{align}
        V^{-1} = \frac{1}{\det{V}} \mat{Q_{22} & -Q_{12} \\ -Q_{12} & Q_{11} + q}
            ,\quad
        \vb{\alpha} &= \frac{1}{\det{V}} \mat{\det{Q} & 0 \\ q Q_{12} & \det{V}} \vb{\delta}
            ,\qq{and}
        D = \frac{q \det{Q}}{\det{V}} \mat{1&0\\0&0}
            .
    \end{align}
    Notice that $D$ is a projector onto the first coordinate. Therefore, the two dimensional Gaussian involving $\vb{\delta}$ actually reduces to a one dimensional Gaussian exclusively in the first coordinate.
    \begin{equation}
        G(\vb{0}; \vb{\delta}, D) = G(0; \delta_1, D_{11}).
    \end{equation}
    With these \cref{lemma:special integral} now follows directly from \cref{eq:obfuscated result}.
\end{proof}

\end{document}